\documentclass[11pt]{article}

\usepackage[linesnumbered,vlined,ruled]{algorithm2e}
\usepackage{multirow,amsmath,amsfonts}

\usepackage{amsmath}
\usepackage{amssymb}
\usepackage{latexsym}
\usepackage{epsfig}
\usepackage{array}
\usepackage{color}
\usepackage{enumerate}
\usepackage{fullpage}
\usepackage[margin=1in]{geometry}

\newtheorem{example}{Example}
\newtheorem{definition}{Definition}
\newtheorem{observation}{Observation}
\newtheorem{theorem}{Theorem}
\newtheorem{lemma}{Lemma}
\newtheorem{proposition}{Proposition}

\newcommand{\G}{{\ensuremath {\mathcal G}}}

\def\hyph{\nobreakdash-\hspace{0pt}\relax}

\newlength{\boxwidth}
\makeatletter
\DeclareRobustCommand{\qed}{%
  \ifmmode %
  \else \leavevmode\unskip\penalty9999 \hbox{}\nobreak\hfill
  \fi
  \quad\hbox{\qedsymbol}}
\newcommand{\openbox}{\leavevmode
  \hbox to.77778em{%
  \hfil\vrule \vbox to.675em{\hrule width.6em\vfil\hrule}%
  \vrule\hfil}}
\newcommand{\qedsymbol}{\openbox}

\newenvironment{proof}[1][\proofname]{\par \normalfont
  \topsep6\p@\@plus6\p@ \trivlist
  \item[\hskip\labelsep\bfseries\itshape #1.]\ignorespaces }{%
  \qed\endtrivlist }

\newcommand{\proofname}{Proof}

\newcommand{\score}{s} 
\newcommand{\prob}{p}  

\newcommand{\hide}[1]{}

\begin{document}

\title{Ranking Games that have Competitiveness-based Strategies\thanks{This work was
supported by EPSRC grant EP/G069239/1 ``Efficient Decentralised
Approaches in Algorithmic Game Theory.'' A  preliminary version of this paper appeared as~\cite{ggkv-ec}.}}

\author{Leslie Ann Goldberg$^1$, Paul W. Goldberg$^1$, Piotr Krysta$^1$, Carmine Ventre$^2$\\ \\ 
$^1$ Dept.\ of Computer Science \\
University of Liverpool\\
Ashton Street, Liverpool L69 3BX, U.\,K.\\ \\
$^2$ School of Computing \\
Tesside University \\
Borough Road, Middlesbrough, TS1 3BA, U.\,K.
}

\date{}

\maketitle

\begin{abstract}
An extensive literature in economics and social science addresses {\em contests}, in which
players compete to outperform each other on some measurable criterion, often referred
to as a player's {\em score}, or {\em output}. Players incur costs that are an increasing function of score,
but receive {\em prizes} for obtaining higher score than their competitors.
In this paper we study finite games that are discretized contests, and the problems of
computing exact and approximate Nash equilibria. Our motivation is the worst-case hardness
of Nash equilibrium computation, and the resulting interest in important classes of games that
admit polynomial-time algorithms.
For games that have a tie-breaking rule for players' scores, we present a
polynomial-time algorithm for computing an exact equilibrium in the
2-player case, and for multiple players, a characterization of Nash equilibria that shows
an interesting parallel between these games and unrestricted $2$-player
games in normal form. When ties are allowed, via a reduction from these
games to a subclass of anonymous games, we give approximation schemes for two special cases:
constant-sized set of strategies, and constant number of players.
\end{abstract}

\section{Introduction}

We consider a class of finite games, and the problem of computing their exact and approximate
Nash equilibria. In these games, each pure strategy of a player corresponds to a {\em score},
i.e., a level of attainment (on some measurable criterion).
Each score also has an associate cost of effort (which we assume here is player-specific),
where a higher scores requires higher effort.
If all players select strategies, then their payoffs are obtained as follows.
The players are ranked according to the scores\footnote{We are following Siegel~\cite{Siegel}
in using ``score'' to refer to this quantity; it is also called ``output'' in labor-market
contests~(\cite{LR} and subsequent papers), or the ``bid''~\cite{MS}.}
they selected,
and prizes are awarded to players according to their position in the ranking.
The overall payoff to a player is the value of the prize he wins, minus the
cost of effort for the strategy he chose.
Consequently the players face a trade-off between the cost of a strategy, and its effectiveness
at winning prizes. We call these games \emph{competitiveness-based ranking games}.

By way of illustration, consider a set of athletes who are training for a race. Each athlete
spends time and effort in training, and has an increasing function that maps this upfront cost
to performance (speed in the race). His total utility is the value of the prize that he wins,
minus the cost of making this initial effort. Note that the prize is awarded based on his speed
relative to other competitors, with no consideration given to his speed taken in isolation.

\subsection{Motivation}

For unrestricted normal-form games, the PPAD-completeness results of~\cite{DGP,CDT}
suggest that a Nash equilibrium is hard to compute in the worst case.
Faced with a worst-case hardness result, there are two general routes to computationally
positive results. We may move to the problem of computing a weaker solution concept,
such as {\em approximate} Nash equilibrium (defined in detail in Section~\ref{sec:defs})
where for some $\epsilon>0$ a player's incentive to change strategy is at most $\epsilon$.
For this direction, progress has been rather limited.\footnote{For polynomial-time algorithms
the lowest $\epsilon$ that has been achieved for {\em bimatrix} games is just over $\frac{1}{3}$~\cite{TS08},
and for {\em well-supported} approximate equilibria (for which there is a further
constraint that no positive probability may be allocated to any pure strategy that is
worse than the best response by more than $\epsilon$) it is just below $\frac{2}{3}$~\cite{FGSS12}.
For normal-form games with more than 2 players, known results are even weaker~\cite{BGR,HRS}.
These $\epsilon$-values assume all payoffs lie in the range $[0,1]$, so that achieving $\epsilon=1$
is trivial, while $\epsilon=0$ corresponds to an exact equilibrium.}
The main open problem of this line research is to determine the existence of a
{\em polynomial-time approximation scheme} (PTAS); it is known from~\cite{CDT} that
a {\em fully} polynomial-time approximation scheme (FPTAS) is as hard as computing an
equilibrium exactly.
The alternative route to positive results is to note that the PPAD-hardness of computing
a Nash equilibrium relies on a class of highly-artificial games, suggesting that we instead focus
on special cases that represent important and well-motivated games.
When we find that (exact or approximate) Nash equilibria can indeed be efficiently
computed for such a class of games, this overcomes the complexity-theoretic
objection to using Nash equilibrium as a solution concept.
Various classes of potential games, games on congestion networks, and games on graphs (representing
social networks, or networks with local interactivity) do indeed have efficiently computable
Nash equilibria, see~\cite{AGTbook} for an overview.
The contests for prizes that we study here have a rich economics literature, but we believe
that this is the first paper to analyze the problem of {\em computing} their Nash equilibria.
While we find that certain algorithms for anonymous and polymatrix games are
applicable to special cases, we also exhibit novel
polynomial-time algorithms and approximation schemes for games in this class.

The kind of games that we study here are often constructed by a competition organizer
(for example, sporting contests, or reward structures in organizations).
Hence there is an associated mechanism design
problem of structuring the contest in such a way as to elicit competitive behavior from
the participants. Efficient computation of Nash equilibria for these games should help
with the problem of determining, for example, a good allocation of prizes in such a contest.
We noted that the strategies available to players have associated levels of effort and
score, so in a given Nash equilibrium it is straightforward to compute, for example,
the total effort or score of the players. In this way we have a well-defined measure
of goodness of a Nash equilibrium.

Notice that for the players, it is socially optimal to use strategies having minimal effort,
since all prizes must be allocated, irrespective of what the players have achieved
(it is relative, not absolute, achievement that gets rewarded).
However, in a Nash equilibrium, players will typically be more than minimally
competitive. Competitiveness amongst the players  results
in a positive externality --- in the context of spectator sports, the
spectators prefer to see a well-run race, or in the context of research contests~\cite{CG03}
(such as the DARPA Grand Challenge), competitiveness leads to research progress.
We continue by reviewing the background literature in more detail.

\subsection{Related work}\label{sec:related}

As we noted, there is an extensive literature on {\em contests} where
players exert effort with the aim of outperforming their rivals.
The new aspect of this work is our focus on the design of
{\em polynomial-time} algorithms for computing outcomes (equilibria) of these games.
Most previous work considers continuous games, where players have a continuum
of actions to choose from, with real-valued functions from effort
to score. Here we study a discretized version so that we have finite games.
These finite games can still serve as approximations to the continuous ones,
and the resulting discretized functions benefit from a natural representation
(as a list of pairs of values associating effort cost with score, for each player)
so that we have a clear notion of ``input
size'' of a problem instance, as needed in the context of polynomial-time
algorithms.

The most closely-related class of finite games in the literature, appears
to be the {\em ranking games} of Brandt et al.~\cite{BFHS}.
In such a game, the outcome of any pure-strategy profile is a rank-ordering
of the players, and a player's utility is a decreasing function of his position
in the rank-ordering. However, the games studied in~\cite{BFHS} allow an arbitrary dependence of
rankings on the pure-strategy profiles that may cause them. There is no
requirement that certain pure strategies are more likely than others to
raise a player's ranking. This leads to computational hardness results, notably that
unrestricted ranking games are hard, even for just 3 players~\cite{BFHS}.

{\em Anonymous games} represent a class of finite games that relate to the
discretized contests considered here. Anonymous games are games where a player's
payoff depends on his own action and on the distribution of actions taken by the
other players, but not on the identities of the players who chose each action.
Anonymous games admit polynomial-time approximation schemes
(PTAS's)~\cite{Dask,DP} but may be PPAD-complete to solve exactly.
The algorithms for anonymous games can be applied to an interesting subclass of the
discretized contests that we study here.
In particular they apply to a special case in which all players have
the same (finite) set of score levels available to them,
with prizes being shared in the event of ties (which is a standard assumption in much of
the literature).

We next mention some of the more classical literature on continuous contests.
An influential line of work is the literature on rent-seeking problems,
initiated by Tullock~\cite{T}. These are problems in which players compete
to receive favorable treatment from a regulator. Another large body of
literature initiated by Lazear and Rosen~\cite{LR} has focused on contests in labor markets.
Lazear and Rosen study the merits of rank-based prizes (as an alternative
to paying a piece rate) as a means of incentivizing effort by workers in organizations.
These models incorporate a random noise process that affects the selection of the winner.
In a Tullock contest, the probability of winning is the amount of effort
exerted by a player, divided by the total effort.
Besides being a model of artificial competition, games of this kind are a model
for competition for status within society~\cite{HK}.

Closer to the setting of this paper, are contests where the outcome
is a deterministic function of effort spent by the players.
Siegel~\cite{Siegel} studies properties of the Nash equilibria of (continuous) contests,
in a setting where all prizes have the same value, so that a contestant either wins or loses.
He analyses expected payoffs to players,
and also {\em participation} (which refers to the decision by a player to make more
than a minimal expected effort).
An important special case is first-price all-pay auctions, where an item (the prize) is sold
to the bidder who makes the highest offer, but all bidders must pay, even if they lose.
In an all-pay auction, there is a sense in which (cost of) effort is the same thing as score.
Where players have different valuations for the prize, Baye et al.~\cite{BKdV} show that
this can be interpreted as score being some linear function of effort, with scaling factor
proportional to the valuation.
For the case of two players, Hillman and Riley~\cite{HR} study the cumulative
probability distributions of effort choices in the unique Nash equilibrium of
the game; the uniqueness of the equilibrium is proved in~\cite{BKdV}.
The latter paper also extends the analysis to the case of first-price all-pay
auctions with many players. A survey of results and models considered in contests
(without noise) is contained in~\cite{K09}. In general, our games can be viewed
as a discretized version of contests without noise and could be used as an
approximation to these games, provided the discretization is fine enough.

As we noted, this model has a corresponding mechanism design problem of allocating
values to a set of prizes in a contest (so as to maximize total effort or score by the
players). A well-known paper of Moldovanu and Sela~\cite{MS}
analyses this question in a setting where the players' functions (from effort
to score) differ from each other in being linearly scaled by each player's ability.
They obtain results for allocating value to prizes, given prior distributions
over the players' abilities. (\cite{MS} also begins with an informative and
readable discussion and motivation for the study of contests.)
Moldovanu et al.~\cite{MSS07} study how to elicit maximal effort via the selection of rank-based
{\em status classes}, where a player gains utility from having others assigned to lower classes,
and disutility from others in higher classes. (In~\cite{MSS07} each player has a privately-known
ability, relating score to cost of effort, which has been generated by a probability distribution.
The solution concept is a shared function mapping ability to effort.)
Their model relates to our result for {\em linear-prize games} (Section~\ref{sec:linear}),
in which the value of taking $k$-th place in the ranking is linear (decreasing) as
a function of $k$; that essentially corresponds to $d$ status classes each of size 1, $d$ denoting the number of players.
Note that the setting we study here is different, in that players have arbitrary
discretized functions from effort to score, and the functions are all commonly known.
An alternative approach to maximizing effort~\cite{MS06} considers how to
divide the competitors into sub-contests whose winners then compete in a final round.

Szymanski~\cite{Szymanski} considers the application of the theory of contests
to the design of sporting competitions.
The allocation of prizes in dynamic sport contests in which players determine
their efforts at different stages of the game (e.g., at the beginning of each
half in a soccer game) is considered in~\cite{CCH}. There, the focus is on
comparing rank-based versus score-based prizes when spectators care about
contestants' efforts or about the ``suspense'' of the game (see also~\cite{MS,CCH}).
Cohen et al.~\cite{Cohen} study a version of the contest design problem where
the prize fund may be chosen by the designer, who wants to maximize the effort
elicited, with prizes representing the price paid for effort.
A related line of research~\cite{Haugen, BC} has addressed from a
game-theoretic perspective, the impact of the point scoring system on
offensive versus defensive play in the UK Premier League; our concern
here is slightly different, being focused on highly competitive versus
weakly competitive play.


\subsection{Our Contribution}

Some of our algorithms apply specifically to an interesting
special case of {\em games without ties} where players cannot
share a prize as a result of obtaining the same score.\footnote{As an example of this consider
the competition amongst universities for places in a ranking, or league table. In such a ranking,
it is necessary to list the names of institutions in a strict order; ties have to be broken somehow,
and this is often done alphabetically by name of institution.}
For games without ties, we give an efficient algorithm for the 2-player case, and for the
multi-player case we show how to compute the probabilities in a Nash equilibrium if the
supports of the players' distributions are known (Theorem~\ref{thm:supportpoly}).
(This shows an interesting parallel with general 2-player bimatrix games, where a Nash equilibrium
having known support can be efficiently computed.)
When ties are possible ---in the literature, the
standard assumption is that prizes are shared--- it is convenient to reduce these games to
equivalent {\em score-symmetric} games (Definition~\ref{def:ret-symm}) in which all
players have the same available set of score values, but with player-dependent costs.
The reduction incurs an increase in the number of strategies that is proportional to the
number of players. For these games, a special case of interest arises when we assume
a constant limit on the number of these strategies, and the results of ~\cite{Dask,DP} can
be used to provide a PTAS; here we give a simpler PTAS for such games.
Table~\ref{tbl:algoresults} gives our algorithmic results for both classes of
games, together with {\em linear-prize games} (Section~\ref{sec:linear}), a class of ranking
games in which the prize for taking position $k$ in the ranking is a linear function of $k$.

\begin{table*}[t]
\centering
\begin{tabular}{p{80pt}|c|c|c|c|}
\cline{2-5}
& \# players & \# prizes & \# actions & \rule{0ex}{15pt} Result\\[5pt]
\hline
 \multicolumn{1}{|c|}{\multirow{3}{80pt}{\centering Score-symmetric games}}  &   any \rule{0ex}{12pt} & any  & O(1) & PTAS (Thm. \ref{thm:ptasfixedstrat})\\
\cline{2-5}
 \multicolumn{1}{|c|}{} & O(1)\rule{0ex}{12pt} & O(1) & any & FPTAS (Thm. \ref{thm:fptas2})\\
\cline{2-5}
 \multicolumn{1}{|c|}{} & any \rule{0ex}{12pt} & any  & 2 & Exact Pure (Thm. \ref{thm:retsym:pne})\\
\hline
\multicolumn{1}{|c|}{\centering Games without ties} & 2\rule{0ex}{12pt} & 2 & any & Exact (Thm.~\ref{thm:2player-no-ties})\\
\hline
\multicolumn{1}{|c|}{\centering Linear-prize games} & any\rule{0ex}{12pt} & \# players & any & Exact (Thm.~\ref{thm:linear}) \\
\hline
\end{tabular}
\caption{Our algorithmic contributions.} \label{tbl:algoresults}
\end{table*}

\section{Model, notation and some illustrative examples}

We work in a classical game-theoretic setting of a finite number of players, each with
a finite number of actions, and we consider the problem of computing Nash equilibria, and
approximate Nash equilibria, for these games.

In Section~\ref{sec:notation} we specify in detail the class of games that we study,
and introduce some notation and terminology. Section~\ref{sec:defs}
gives the background definitions of Nash and approximate Nash
equilibrium. Section~\ref{sec:example} shows some examples to illustrate
various technical issues.

\subsection{Definition and Notation}\label{sec:notation}

A {\em prize} refers to the reward that a player gains from
obtaining a specified position in the ranking, and this relates directly to the standard
usage of ``first prize'', ``second prize'' etc in competitions.
We say that an action is ``stronger'' or ``more competitive'' than another
one, if its score is higher. Any pair of actions are comparable
in this sense, whether or not they belong to the same player.

We next formally define the class of games we introduce and study in this paper.
Throughout, we let $d$ denote the number of players in a game, and $n$ the number
of strategies available to each player.

\begin{definition}
In a competitiveness-based ranking game, the $j$-th pure strategy of player $i$,
denoted $a^i_j$, has associated a {cost} $c^i_j$ and a {score} $\score^i_j$.
We assume they are indexed in increasing order of competitiveness so that,
for all $i$, $j$, we have $c^i_j < c^i_{j+1}$ and $\score^i_j < \score^i_{j+1}$.
Any pure-strategy profile results in a ranking of the players according to the scores.
A player whose position in the ranking is $k$ gets awarded the $k$-th prize, having value $u_k$.
Prizes are non-increasing with respect to ranks: $u_k \geq u_{k+1}$, for $1\leq k < d$,
with the assumption that $u_1 > u_d$.\footnote{This assumption simply rules out an
uninteresting case in which competition is not adequately incentivised.
Indeed, when $u_1=u_d$, the profile in which all players play their least
competitive action is a dominant strategy equilibrium.} In the event of a tie (where two or
more players obtain the same score and are ranked equal) the prizes
that would result from tie-breaking are shared.
The total payoff to a player will be the value of the prize he is awarded,
minus the cost of the action selected by that player.
\end{definition}

For the purpose of designing algorithms that search for a Nash equilibrium, we can assume
without loss of generality that we have \emph{strict} monotonicity of effort costs and scores,
in the definition above. If two different actions
(i) have the same cost and different scores then the stronger dominates the weaker and
(ii) have the same score and different costs then the cheaper will dominate the more expensive.

Next we define an interesting subclass of the games we consider.

\begin{definition}\label{def:ret-symm}
A \emph{score-symmetric} game is a competitiveness-based ranking game in which all players have the
same set of pure strategies (which we denote $a_1,\ldots,a_n$) having the same scores
(which we denote $\score_1,\ldots,\score_n)$.
Costs remain player-specific, and $c^i_j$ denotes the cost to player $i$ of pure strategy $a_j$.
\end{definition}

\subsection{Exact and approximate Nash equilibria}\label{sec:defs}

Here we give the definitions of Nash equilibrium and approximate Nash equilibrium,
also some further notation we use throughout.
Let $S_i$ be the set of player $i$'s pure strategies; $S_i=\{a^i_j\}_j$.
Let $S=S_1\times\ldots\times S_d$ be the set of pure-strategy profiles,
where recall $d$ denotes the number of players. It is convenient to define
$S_{-i}=S_1\times\ldots\times S_{i-1} \times S_{i+1} \times\ldots\times S_d$
as the set of pure-strategy profiles of all players but $i$.

A {\em mixed strategy} for player $i$ is a distribution on $S_i$, that is,
real numbers $x^i_j\geq0$ for each strategy $a^i_j\in S_i$ such that $\sum_{j\in S_i}x^i_{j}=1$.
A set of $d$ mixed strategies (one for each player)
is a {\em mixed strategy profile}. By $u^i_s$ we denote the utility to player $i$ in
strategy profile $s$. A mixed strategy profile $\{x^i_j\}_{j\in S_i}, i=1,\ldots,d$,
is called a {\em (mixed) Nash equilibrium}\ if, for each $i$, $\sum_{s\in S} u^i_{s} x_s$ is maximized
over all mixed strategies of $i$ ---where for a strategy profile $s=(s_1,\ldots,s_d)\in S$,
we denote by $x_s$ the product $x^1_{s_1}\cdot x^2_{s_2} \cdots x^d_{s_d}$.
(The notation $x_s$ naturally extends to strategy profiles $s \in S_{-i}$.) That is, a Nash
equilibrium is a set of mixed strategies from which no player has a
incentive to unilaterally deviate.  It is well-known (see, e.g.,~\cite{OR}) that the following is
an equivalent condition for a set of mixed strategies to be a Nash equilibrium:
\begin{align}\sum_{s\in S_{-i}} u^i_{js} x_s > \sum_{s\in S_{-i}}
u^i_{j's} x_s
  \Longrightarrow x^i_{j'}=0.\label{eq:NEconstraints}\end{align}
The summation $\sum_{s\in S_{-i}} u^i_{js} x_s$ in the above equation is
the expected utility of player $i$ if $i$ plays pure strategy $j$
and the other players use the mixed strategies $\{x^{i'}_k\}_{k\in S_{i'}}, i'
\neq i$.  Nash's theorem~\cite{N} asserts that {\em every game has a
Nash equilibrium}.

We say that a set of mixed strategies $x$ is an {\em
$\epsilon$-approximately well supported Nash equilibrium}, or {\em
$\epsilon$-Nash equilibrium} for short,  if, for each $i$, the following
holds:
\begin{align}
\sum_{s\in S_{-i}} u^i_{js} x_s > \sum_{s\in S_{-i}} u^i_{j's} x_s
+\epsilon
  \Longrightarrow x^i_{j'}=0. \label{eq:epsilonNEconstraints}
\end{align}
Condition (\ref{eq:epsilonNEconstraints}) relaxes that in
(\ref{eq:NEconstraints}) by allowing a strategy to have positive
probability in the presence of another strategy whose expected
payoff is better by at most $\epsilon$.

\subsection{Some examples}\label{sec:example}

We consider some examples that should be helpful in understanding the
model and issues arising. Example~\ref{ex:2player} shows that the games we
consider do not always have {\em pure} Nash equilibria.

\begin{example}\label{ex:2player}
Consider two players; for $i=1,2$ player $i$ has two actions $a^i_1$ and $a^i_2$.
Suppose the row player (player 1)
is stronger than the column player in the sense that the column player
only wins by playing $a^2_2$ while the row player plays $a^1_1$ --- this
can be achieved by setting $\score^2_1=2$, $\score^1_1=3$, $\score^2_2=4$, $\score^1_2=5$.
Suppose the costs are $c^i_1=0$, $c^i_2=\frac{1}{2}$ for both players $i=1,2$, and
we have a single prize worth 1, i.e., $u_1=1$ and $u_2=0$. We have payoff matrix:
\[
\begin{array}{r|cc}
          & a^2_1  &  a^2_2  \\  \hline
  a^1_1   & (1,0)  &  (0,\frac{1}{2})  \\
  a^1_2   & (\frac{1}{2},0)  &  (\frac{1}{2},-\frac{1}{2})
\end{array}
\]
It is easily checked that this game has no pure Nash equilibrium and that the unique
equilibrium is the one in which both players mix uniformly.
\end{example}

Example~\ref{ex:manyplayers} is an anonymous game with binary actions
(studied in~\cite{Blonski}, although~\cite{Blonski} studies a continuum of players).
The example shows that in this kind of game,
there may be multiple equilibria, and the number of equilibria may be exponential in
the number of players.

\begin{example}\label{ex:manyplayers}
Consider a symmetric game with an even number $d\geq 4$ of players; a single
prize worth 1 unit; each player $i$ has two actions $a_1$ and $a_2$
with costs $c_1=0$ and $c_2=c$ (we do not have a superscript to identify
a player, since the games are symmetric). The prize will be shared
between players who use $a_2$, or all players if they all use $a_1$.

Notice first that for $c\in \left(0,\frac{1}{d}\right]$, there is a pure equilibrium in which all players play $a_2$.
For $c\in \left(\frac{1}{d},1\right)$ there is, by symmetry, a fully-mixed
Nash equilibrium where all players play $a_2$ with the same probability.
This can be seen by the following argument. Suppose each player, other than the first,
plays $a_2$ with probability $p$. The first player has an incentive to play $a_2$ which is
decreasing in $p$. In particular, for $p=1$ player $1$ has no incentive to play $a_2$,
while for $p=0$ player $1$ has an incentive to play $a_2$. Then, by continuity, there exists a
value of $p$, say $p^*$, for which player $1$ is indifferent between $a_1$ and $a_2$.
In a profile in which all players play $a_2$ with probability $p^*$ all players are indifferent by symmetry.

Now put $c=\frac{2}{d}-\epsilon$, where $\epsilon < \frac{2}{d^2+2d}$.
We claim that there are also {\em pure} Nash
equilibria where any subset of size $\frac{d}{2}$ play pure $a_2$ and
the others play pure $a_1$. A player playing $a_2$ obtains utility
$-c+\frac{2}{d}>0$; no incentive to switch to $a_1$. A player playing
$a_1$ obtains utility 0, and by switching to $a_2$ would obtain utility
$-c+1/(\frac{d}{2}+1) = \epsilon-\frac{2}{d} + \frac{2}{d+2} < 0$.
Indeed there are also many mixed equilibria where a subset of the
players play pure $a_1$ and the other players all use the same
probabilities.

Observe that there are no Nash equilibria where players may mix with
different probabilities --- two such players would both be indifferent
between $a_1$ and $a_2$, but their expected payoffs from playing $a_2$
would have to differ.
\end{example}

The following example shows that there is no bound on the
price of anarchy and on the price of stability in these games.

\begin{example}
Consider a symmetric game with 2 players; a single prize worth 1 unit;
each player $i$ has two actions $a_1$ and $a_2$ with scores $\score_1 < \score_2$
and costs $c_1=0$ and $c_2=1/2 - \epsilon$,
for some small $\epsilon > 0$. The payoff for playing $(a_2,a_2)$ is $\epsilon$ for both players,
and it is higher than the payoffs obtained by deviating from the strategy $a_2$: player $1$ has
a payoff of $0$ for strategy profile $(a_1,a_2)$ and so does player $2$ for strategy profile $(a_2,a_1)$.
Thus, strategy profile $(a_2,a_2)$ is a pure Nash equilibrium and its social welfare is $2\epsilon$.
Now, notice that $a_2$ is actually a strictly dominant strategy for both players thus implying that
no other action profile is a Nash equilibrium. The action profile that
maximizes the social welfare is $(a_1,a_1)$, and its social welfare is $1$.
Thus, the price of anarchy in this game is $1/(2\epsilon)$.
Since this Nash equilibrium is unique, $1/(2\epsilon)$ is also the price of stability of this game.
Because $\epsilon$ can be chosen arbitrarily small, both price of anarchy and price of stability are
unbounded. Note that this game is essentially the Prisoner's Dilemma in which $a_1$ is the collaborating strategy and $a_2$ is the defecting one.
\end{example}

\section{Algorithms and proofs}

We start by noting some preprocessing steps that establish some useful
assumptions that we can make without loss of generality. We continue in
Section~\ref{sec:noties} by considering separately the special case
where players cannot tie for a position in the ranking; this case would
arise in competitions that have a tie-breaking rule, or where the score values
$\score^i_j$ are all distinct.
The reason for a focus on the tie-free case is that the analysis is
simpler and the Nash equilibria turn out to have a special structure.
Section~\ref{sec:linear} applies a result of~\cite{DPpolym} for poly matrix games, to the special case where prize values decrease linearly as a function of rank position.

In Section~\ref{sec:ret-symm} we study the more general case where
players may tie for a position in the ranking. We show that we can
focus without loss of generality on an anonymous subclass of these
games. Pure Nash equilibria of these games are studied in Section
\ref{sec:retsym:pne}. In Section
\ref{sec:ptas:fixedstrategies} we give a polynomial-time approximation
scheme for games with a fixed number of strategies.
Finally, in Section~\ref{sec:fptas} we give a fully polynomial-time
approximation scheme for the case of constantly-many players.

\subsection{Preprocessing}\label{sec:preprocessing}

Results about the computation of $\epsilon$-approximate equilibria require
us to assume that all payoffs in a given game lie in some fixed bounded range;
usually the interval $[0,1]$ is assumed. Games whose values lie outside this
this range can be have their payoffs resealed into $[0,1]$ without affected the
strategic aspects of the game. With this in mind, we resale the payoffs of
an arbitrary competitiveness-based ranking game as follows.
We may assume that the number of prizes is equal to~$d$, the number of players.
This is without loss of generality --- if there are more prizes than players, then
only $d$ of them can be awarded (so all but the first $d$ of them can be discarded
without changing the game). Also, if there are fewer prizes than players,
we can just add additional prizes of value~$0$.
(Thus, when we say that we have a game with only $j$ prizes, what we really
mean is that $u_{j+1}= \cdots = u_d = 0$.)
The first step in the preprocessing is to ensure that $u_d=0$.
This can be done, without changing the strategic aspect of the game, by subtracting $u_d$ from all prizes.
Next, we ensure that $u_1=1$. This can be done, without changing
the strategic aspect of the game, by dividing all prizes, and all costs, by $u_1$.
Next, for each player~$i$, we ensure that $c_1^i=0$.
This can be done, without changing the strategic aspect of the game, by subtracting $c_1^i$
from all of the costs of player~$i$.
Finally, we may assume that no player has an action with a cost greater than $1$,
since such an action would be dominated by $a^i_1$.

Note also that the numerical values of the scores $\score^i_j$ may be
modified without affecting the payoffs and Nash equilibria of the game,
provided only that the modification does not affect which are greater
than which (in which case the ranking of the players is preserved).
However, it is usually convenient to specify numerical $\score^i_j$ values
(rather than, more abstractly, their ordinal relationships) when describing a game.

Finally, we establish a useful fact that will be used to obtain
polynomial-time algorithms that return approximate Nash equilibria.

\begin{observation}\label{obs:roundprob}
For any $\epsilon>0$ inverse of an integer, given a probability vector
$\mathbf{x}=(x_1,\ldots,x_n)$, it is possible to define a probability vector
$\tilde{\bf x}=(\tilde x_1,\ldots,\tilde x_n)$, called an ``$\epsilon$-rounding''
of $\bf x$, in which
\begin{enumerate}
\item each entry is equal to a non-negative integer multiple of $\epsilon$, and
\item For every $j\in \{1,\ldots,n\}$, the rounding error
$\left|\sum_{k=1}^j (\tilde{x}_k - x_k) \right|$ is less than~$\epsilon$.
\end{enumerate}
\end{observation}
\begin{proof}
We consider the values~$x_1,\ldots,x_n$ in order.
When we consider~$x_j$, we round it to define $\tilde x_j$.
If $x_j$ is an integer multiple of $\epsilon$, then $\tilde x_j=x_j$.
Otherwise, we set the value of $\tilde x_j$ by rounding $x_j$ ---
we round up to the nearest integer multiple of~$\epsilon$
if $\sum_{k<j} \tilde x_k \leq \sum_{k<j} x_k$ and
we round down to the nearest integer multiple of~$\epsilon$ otherwise.
This ensures that $\left|\sum_{k=1}^j (\tilde{x}_k - x_k) \right|< \epsilon$.

We now show that $\tilde{\bf x}$
is a probability vector, that is, that $\sum_{k=1}^n \tilde{x}_k = 1$.
Now since $\bf x$ is a probability vector,
$$\left|
\left(\sum_{k=1}^n \tilde{x}_k\right)-1
\right| =
\left|\sum_{k=1}^n (\tilde{x}_k - x_k) \right|,$$
and we already know that the latter is less than~$\epsilon$.
Since the $\tilde{x}_k$'s are non-negative integer multiples of $\epsilon$ then
$\sum_k \tilde{x}_k = a \epsilon$, for some integer $a$.
The above inequality then yields $-1 < a - \frac1\epsilon < 1$.
Since $\frac1\epsilon$ and $a$ are integer numbers, we can only satisfy
the previous inequality by having $\sum_k \tilde{x}_k=1$.
\end{proof}

\subsection{Games without ties with a single prize}\label{sec:noties}

We begin by observing that (subject to the above preprocessing) we may restrict our
attention to actions' costs that are \emph{strictly} less than $1$.

\begin{observation}
Assume player $i$ has an action $a^i_j$ such that $c^i_j=1$. Then $a^i_j$
is \emph{weakly} dominated by $a^i_1$. Therefore, we can eliminate
$a^i_j$ from the game at the price of eliminating some potential Nash
equilibria.
\end{observation}

Assuming that costs are strictly less than $1$, we show that Nash
equilibria of games without ties have a nice structure when there is a
single prize. Siegel~\cite{Siegel} has shown a more general version of the
following, in the context of continuous games. We include a proof here,
since it is simpler in the discrete case.

\begin{theorem}\label{thm:onepositive}
Suppose there is a single prize of value 1 and actions' costs are less than $1$.
If no two actions have the same strength (thus ties are impossible) then in any Nash equilibrium
\begin{enumerate}
\item There is just one player with positive expected payoff; all others have expected payoff zero.
\item The player with positive expected payoff is the one with the strongest action with a cost of less than 1.
\end{enumerate}
\end{theorem}

\begin{proof}
Given the preprocessing steps noted above and the assumption that the single prize has value $1$,
the costs of all actions lie in the range $[0,1)$, and each player has an action with cost 0.
Let ${\cal N}$ be a Nash equilibrium. For each player $i$ let $w_i$ be the
weakest action of $i$ that lies in the support of ${\cal N}$; thus, $i$
has positive probability of using $w_i$, and all other actions that $i$
uses with positive probability are stronger than $w_i$.

For any action $a$, let $\score(a)$ denote the score of that action.
Let $p$ be the player whose
weakest action in the support of ${\cal N}$ is stronger than all other players' weakest
actions in the support of ${\cal N}$, thus $\score(w_p) > \score(w_{p'})$ for all $p'\not= p$.

Note that for any player $p'\not= p$, the expected payoff to $p'$ from using
action $w_{p'}$ is non-positive: $w_{p'}$ cannot win since $p$ is certain to
play a stronger action. But, $p'$ gives positive probability to $w_{p'}$,
so no other action available to $p'$ can have higher expected payoff.
$p'$ has non-positive expected payoff, and under the assumption (that
we may adopt from preprocessing) that players
all have a 0-cost action, $p'$'s expected payoff must in fact be zero.

For the second part of the theorem, note that we have seen that all but one player must have
expected payoff 0, in Nash equilibrium ${\cal N}$.
Let $p''$ be the player with the strongest action having a cost of less than 1.
Then $p''$ can always guarantee a positive expected payoff by using that action, so in any
Nash equilibrium, must receive positive expected payoff. Hence $p''$ is the only player with
non-zero expected payoff in ${\cal N}$ (and $p''$ is in fact the same as the above player $p$).
\end{proof}

The following theorem also shows how a Nash equilibrium may be efficiently computed
for tie-free games, provided that we know the support of the Nash equilibrium. This shows
an interesting parallel between these games, and general 2-player normal
form games, especially in conjunction with the subsequent observation that
the solution is a rational number.

\begin{theorem}\label{thm:supportpoly}
For games with any number of players, pure-strategy costs less than $1$ and a single prize
where ties are impossible, a Nash equilibrium can be computed in polynomial time
if we are given the support of a solution.
\end{theorem}

\begin{proof}
Given a game ${\cal G}$, suppose that we remove the pure strategies that
are not in the support of some (unknown) Nash equilibrium. The resulting
game ${\cal G}'$ has a fully-mixed equilibrium ${\cal N}$, thus any
two strategies that belong to a player have the same expected payoff
in ${\cal N}$. Our general approach is to compute the probabilities
$x^i_j$ in descending order of strength of the associated actions $a^i_j$.

Let $a^i_j$ be the strongest action in ${\cal G}'$ (i.e. having the
highest score). Player $i$'s expected payoff is $1-c^i_j$ and by
Theorem~\ref{thm:onepositive} all other players have expected payoff 0.

Let $a^{i'}_{j'}$ be the second-strongest action in ${\cal G}'$; we may
assume $i'\not= i$ since if $i'=i$ then $a^i_j$ would be strictly
dominated by $a^{i'}_{j'}$. Its expected payoff to $i'$ is $-c^{i'}_{j'}
+ (1-x^i_j)$, which by Theorem~\ref{thm:onepositive} is 0, so we have an
expression for $x^i_j$. Consider the third-strongest action
$a^{i''}_{j''}$, whose payoff is given by $-c^{i''}_{j''} +
(1-x^i_j)(1-x^{i'}_{j'})$ (assuming $i''\not= i$) which gives us an
expression for $x^{i'}_{j'}$.

Generally, the $r$-th strongest action $a^\alpha_\beta$ has expected payoff
$-c^\alpha_\beta + \prod_{k\not=\alpha} (1 - \sigma_{k})$
where $\sigma_{k}$ is the sum of probabilities of player $k$'s actions that are stronger than $a^\alpha_\beta$.

The probabilities for each player's weakest actions will be obtained from
the equations that ensure that for every player $i$, the values
$x^i_j$ sum to 1 (are a probability distribution).
\end{proof}

\begin{observation}\label{thm:rational}
For games where ties are impossible, if all action costs are rational numbers smaller than $1$
then the solution is also a rational number.
\end{observation}

This is immediate from the expressions in the above proof that give the
values $x^i_j$.

\subsubsection{Solving 2-player games exactly}

Ranking games (as in~\cite{BFHS}) with actions that do not have the
upfront costs $c^i_j$ we consider here, are constant-sum, so they can be
solved efficiently in the 2-player case. Our games are not constant-sum, but
we do have an alternative polynomial-time algorithm to solve them in the
2-player case.

\begin{theorem}\label{thm:2player-no-ties}
2-player ranking games that have competitiveness\hyph based strategies and are without ties can be solved
exactly in polynomial time.
\end{theorem}

\begin{proof}
As before, assume a single prize of 1 unit and action costs in $[0,1]$, which
may be assumed by the preprocessing noted earlier.

We can (in polynomial time) compute exact solutions of 2-player games
of this type as follows. We start by eliminating certain dominated
strategies. Specifically, suppose that for strategies $a^i_j$ and $a^i_{j+1}$,
the set of opponent's strategies that they win against, is the same.
Then $a^i_{j+1}$ can be eliminated. Rename the
strategies of this game $(a_1,\ldots a_n)$ for the row player and
$(a'_1,\ldots a'_{n'})$ for the column player. Assume without loss of
generality that it is the row player who has the weakest strategy,
thus $a'_1$ wins against $a_1$. When $n=n'$, the strategies, arranged
in ascending order of strength are $a_1,a'_1,a_2,a'_2,\ldots a_n,a'_n$.
When instead $n \neq n'$, then it must be the case that the same player
has weakest and strongest strategy; the strategies, arranged in ascending
order of strength, are $a_1, a_1', \ldots, a_n, a_n', a_{n+1}$ in this case.

Suppose that in some Nash equilibrium ${\cal N}$  the row player
does not use strategy $a_j$ for some $j>1$ (that is, the player plays
$a_j$ with probability 0). Then the column player does
not use strategy $a'_j$ (which is the cheapest one that wins against $a_j$)
since $a'_j$ would now be dominated by $a'_{j-1}$.
For a similar reason, the row player will not use $a_{j+1}$, the cheapest strategy
that wins against $a'_j$, so the column player will not use $a'_{j+1}$, and so
on. This shows that (in a Nash equilibrium) the strategies in either
player's support must be either a
prefix of the sequence of his strategies or a prefix of all his strategies but the weakest,
with strategies arranged in ascending order of strength.

We can now try to solve for all such supports, since there are
polynomial-many of them. Recall that a 2-player game can be solved
efficiently in polynomial time if we are told the support of a
solution, since it reduces to a linear program (see, for example, page~31 in~\cite{AGTbook}).
\end{proof}

The main property used to show Theorem~\ref{thm:2player-no-ties} above (i.e., if in any solution
a player does not play a certain strategy $s$ then the other one does
not play the strategy that ``just beats'' $s$) breaks down when ties are
allowed. Simply consider two consecutive strategies of player $1$,
$a_i^1$ and $a_{i+1}^1$, and of player $2$, $a_j^2$ and $a_{j+1}^2$ such
that scores of $a_i^1$ and $a_j^2$  (resp. $a_{i+1}^1$ and $a_{j+1}^2$)
are the same. In this case, we can no longer conclude that in any solution if player
$1$ does not play $a_i^1$ then player $2$ does not play $a_{j+1}^2$ as
such a strategy is used not just to beat $a_i^1$ but also to share with
$a_{i+1}^1$. This shows that for games where ties may occur,
we have to use different approaches to obtain polynomial-time algorithms for them.

\subsection{Exact algorithm for linear-prize ranking games}\label{sec:linear}

Consider a $d$-player $n$-strategies-per-player ranking game $\G$ with
competitiveness-based strategies without ties in which the prize for ranking
$k$-th is a linear function $a-k b$, for some values~$a$ and $b$.
We call $\G$ a \emph{linear-prize ranking game}.
We claim that we can represent $\G$ as a poly matrix game \cite{DPpolym}.
A poly matrix game can be represented as a graph: players are the verticals
and a player's payoff depends on the actions of his neighbors.
The edges are $2$-player zero-sum games. Once all players have chosen a strategy,
the payoff of each player is the sum of the payoffs of the games played with his neighbors.
Dashikis and Papadimitriou~\cite{DPpolym} give an algorithm which solves
polymatrix games in polynomial time.

We can express $\G$ as a polymatrix game as follows.
We define a complete $(d+1)$-vertex graph where the additional vertex encodes an
external player $N$, which we call ``nature''.
For players~$i$ and~$i'$ in~$\G$, edge $(i,i')$ is an $n\times n$
$2$-player constant-sum game. (This is a zero-sum game, shifted by a constant ---
this shifting can be accommodated in polymatrix games.)
The game on edge $(i,i')$ ``punishes'' the lowest-ranked player amongst~$i$ and~$i'$.
In particular,  in the matrix of this game,
the entry $(j,j')$ will have payoff $0$ for player $i$ and $-b$ for player
$i'$ if and only if $\score^i_j>\score^{i'}_{j'}$.
Edge $(i,N)$ is an $n\times 1$ $2$-player zero-sum game in which nature ``gains''
what player $i$ is paying in effort minus $a-b$.
In particular, each entry $j$ will have payoff $a-b-c^i_j$ for player $i$
and payoff $c^i_j-a+b$ for nature. Note that once all players of $\G$ have chosen strategies,
the payoff to the player who is ranked $k$-th is $a - k b$ minus his cost.
He loses~$b$ to each of the $k-1$ players who beats him and he gains $a-b$
and loses his cost of effort, to nature.
Thus, using the algorithm in~\cite{DPpolym}, we obtain the following theorem.

\begin{theorem}\label{thm:linear}
There is a polynomial-time algorithm that computes a Nash equilibrium for linear-prize ranking games.
\end{theorem}

\subsection{Games where ties are possible}\label{sec:ret-symm}

We consider a more general situation in which two or players may have pure strategies
having the same score, and if they play those strategies, any prizes are shared.
We begin by showing that we can study without loss of generality Nash
equilibria of competitiveness-based ranking games in which scores are symmetric.
That is, players have a shared set of $n$ pure strategies, $a_1, \ldots, a_n$.
As above, player $i$ has player-dependent costs $c^i_1 < \ldots < c^i_n$, but
the scores are player-independent, and we denote them $\score_1 < \ldots < \score_n$.

\subsubsection{Reduction to score-symmetric games}\label{sec:score-symm}

The reduction preserves Nash equilibria of the original game (and thus
is a Nash homomorphism) and is presented for the case of $2$-player
games. The generalisation to the $d$-player case is straightforward, although the number of
strategies per player would increase by a factor of $d$.

Suppose that player 1 has action $a^1_j$ and player
2 has no action with score equal to $\score^1_j$. We give player 2
a weakly dominated strategy with score $\score^1_j$ --- if $a^2_k$ is the weakest
strategy of player 2 that has higher score than $a^1_j$, give player 2
an additional strategy with cost $c^2_k$ and score $\score^1_j$. If player
2 does not have a stronger strategy than $a^1_j$, give player 2 an
additional strategy with cost 1 and score $\score^1_j$.

We can assume that each player has $n$ strategies $a_1,\ldots,a_n$ with
scores $\score_1,\ldots,\score_n$ and costs $c_1^1,\ldots,c^1_n$ for player 1 and
$c^2_1,\ldots,c^2_n$ for player 2. Suppose we solve this game, and now we have to recover a
solution to the original game before the weakly dominated strategies were
added. To do this, each player just has to replace their usage of
any weakly dominated strategy by the corresponding weakly dominating one.
Let $a_j$ be a weakly dominated action of player 1; $a_{k}$ denotes the
corresponding weakly dominating action, $k>j$ and then $\score_k > \score_j$.
The probability of player 1 playing $a_j$ at equilibrium is
then added to his probability of playing $a_{k}$. This raises the
question of whether player 2 may be given an incentive to deviate as a
consequence. Such an incentive can only concern actions with scores
within the interval $I:=[\score_j, \score_{k}]$ (for actions with score
outside $I$, player 2 has the same payoff when player 1 plays either
$a_j$ or $a_{k}$ since in these cases the ranking of the players do not
change). First, note that if in the Nash equilibrium player 2 plays
actions with scores in $I$ with positive probability then player $1$
plays $a_j$ with probability $0$. This is because player 1 strictly
prefers $a_{k}$ to $a_j$ when player 2 plays actions with score in $I$
and is indifferent between them for the remaining actions. Therefore, we
do not need to redistribute probability mass in this case. On the other
hand, whenever in the Nash equilibrium the actions having score in $I$
are played by player 2 with zero probability then the change of the
probability distribution of player 1 has no effect on player 2; as
observed above, for each action with score outside $I$, player 2 has the
same payoff when player 1 plays either $a_j$ or $a_{k}$ and then he has
no incentive to deviate.

\subsubsection{Score-symmetric games and pure equilibria}\label{sec:retsym:pne}

Unlike games without ties, for which $2$-player $2$-action games might not
possess pure equilibria (see Example \ref{ex:2player}), score-symmetric games
in which players have only $2$ strategies do have pure Nash equilibria
(for any number of players and any number of prizes).

\begin{theorem}\label{thm:retsym:pne}
2-action competitiveness-based score\hyph symmetric ranking games do have pure Nash
equilibria (any number of players; any action costs for individual players).
Furthermore, a pure Nash equilibrium can be found in polynomial time.
\end{theorem}

\begin{proof}
We have 2 pure strategies $a_1$ and $a_2$, where $a_1$ is less competitive; thus
for each player $i$ we have $c^i_1 \leq c^i_2$. Recall that by preprocessing we
may assume that $c^i_1=0$, so that $c^i_2$ is non-negative, for all $i$.
We show how to identify a pure Nash
equilibrium that consists of a (potentially empty) set of players
playing $a_2$, all of whom have a cost for playing $a_2$ lower than the ones
playing $a_1$.

We may assume that the players are indexed in non-decreasing order of their cost of
playing $a_2$, so that for $1\leq i<n$ we have $c^i_2 \leq c^{i+1}_2$.
Now, let ${\cal Z}_i$ denote the pure
profile in which the first $i$ players play $a_2$ and the remaining $d-i$
players play $a_1$.
We claim that if $i$ has an incentive to deviate from ${\cal Z}_{i-1}$ then no
player $i'$ has an incentive to deviate from ${\cal Z}_i$, with $i' < i$.
Indeed, if $i$ has an incentive to deviate from ${\cal Z}_{i-1}$,
then $\frac{u_i+\ldots+u_d}{d-i+1}$, his utility in ${\cal Z}_{i-1}$, is strictly less than
$\frac{u_1+\ldots+u_i}{i}-c^i_2$, his utility in ${\cal Z}_{i}$.
Since, by definition,  $c^{i'}_2 \leq c^i_2$, for $i' < i$, the previous inequality
implies that player $i'$ is better off by sticking to $a_2$ in ${\cal Z}_{i}$.
Therefore, starting from the profile ${\cal Z}_{0}$ in which all players play $a_1$,
we initially check whether player $1$ has an incentive to deviate.
If not, the profile is a pure Nash equilibrium, otherwise we let him deviate
and we have the profile ${\cal Z}_{1}$.
We can reiterate this process for each player $i$ in this ordering until we reach the
break-even point at which the share of the prize obtained from playing $a_2$ goes
down below the cost of the next player in line.
This shows the existence of the claimed pure Nash Equilibrium, and also constitutes
an efficient algorithm for finding it.
\end{proof}

A related result is known for symmetric games in which players have only 2 strategies:
These games always have a pure Nash equilibrium~\cite{CRVW}.
The above theorem concerns games that are anonymous but not symmetric
(as costs are player-specific). However, let us notice that the arguments used in~\cite{BFH}
to prove the result about symmetric games appear to be similar to ours.
It is easy to see that these arguments fail when two players have three
strategies available, as shown by the next example.

\begin{example}
We have 2 players and 3 actions, namely $a_1$, $a_2$ and $a_3$ ordered
increasingly by score, i.e., $\score_1< \score_2< \score_3$. The prizes are
$u_1=1$ and $u_2=0$, $u_3=0$. Costs are $c_1^i=0$ for $i=1,2$,
$c^1_{2}=\frac{2}{3}, c^1_{3}=\frac{4}{5}$ and $c^2_{2}=\frac{1}{3},
c^2_{3}=\frac{2}{3}$. Thus, we have payoff matrix
\[
\begin{array}{r|ccc}
          & a_1  &  a_2 & a_3  \\  \hline
  a_1   & (\frac{1}{2},\frac{1}{2})  &  (0,\frac{2}{3}) &  (0,\frac{1}{3}) \\
  a_2   & (\frac{1}{3},0)  &  (-\frac{1}{6},\frac{1}{6}) &  (-\frac{2}{3},\frac{1}{3}) \\
  a_3   & (\frac{1}{5},0)  &  (\frac{1}{5},-\frac{1}{3}) & (-\frac{3}{10},-\frac{1}{6})
\end{array}
\]
It is easily checked that this game has no pure Nash equilibrium.
The unique Nash equilibrium of the game is $(\frac{2}{3},0,\frac{1}{3})$ for player $1$
and $(\frac{2}{5},\frac{3}{5},0)$ for player $2$.
\end{example}

\subsubsection{PTAS for many players who share a fixed set of strategies}\label{sec:ptas:fixedstrategies}

Consider a score-symmetric $d$-player game, each of whom have
actions $a_1,\ldots,a_n$ with scores $\score_1,\ldots,\score_n$.
In this section we view the number of actions $n$ as a constant, and we
are interested in algorithms whose runtime has polynomial dependence on $d$,
the number of players. A score-symmetric game is a special case of an
anonymous game, so it is possible to directly apply a result of Daskalakis
and Papadimitriou~\cite{DP} to show that it has a PTAS.
Here, we give a conceptually simpler PTAS.

In the PTAS, we first round the cost vectors of the players.
This does not introduce much error, but it ensures that there are only
a constant number of different cost vectors.
We refer to the cost vector of a player as its ``type''.
Now the point is that players of the same type are
equivalent in the following sense --- once we know how many players of each type adopt
each (mixed) strategy, we can examine the resulting strategy profile to check whether it
is an $\epsilon$-Nash equilibrium.
It is only important \emph{how many} players of each type adopt a particular strategy --- it is not important
which  players they are. Thus, brute-force search is quite efficient.
Technically this algorithm is {\em oblivious} in the sense of
Daskalakis and Papadimitriou~\cite{DP09}, in that it constructs a polynomial-sized
set of mixed-strategy profiles in such a way that at least one of them should be
an approximate equilibrium, and checks each of them.

By a {\em $k$-composition} of a positive integer~$N$, we mean a solution to $N_1 + \ldots + N_k = N$
in which $N_1,\ldots,N_k$ are non-negative integers. There are less than $N^k$ such solutions.

\begin{algorithm}[h]\label{ptasRS}

\caption{PTAS for score-symmetric $d$-player games having a constant number, $n$,
of pure strategies and a given accuracy parameter $\epsilon>0$}

Let $\delta=\epsilon/n$ and $\ell =\lceil 1/\epsilon \rceil$.

For each player, round each cost $c^i_j$ down to the nearest non-negative
integer multiple of $\epsilon$. Two players have the same {\em type} if they have
the same set of rounded costs. Let ${\cal T}$ be the set of types; note that
$|{\cal T}|\leq {(\ell+1)}^n$.

\label{stepRS:constantactions}

Let $S$ be the set of $n$-dimensional probability vectors $\{(x_1,\ldots,x_n)\}$
in which each $x_j$ is a non-negative integer multiple of $\delta$.
(Mixed strategies for individual players will be sought from elements of $S$.)
Let $s=|S|$, and note that $s\leq {(n\ell+1)}^n$.

\label{stepRS:all}

Perform a brute-force search as follows.
For every type~$t\in {\cal T}$, consider every $s$-composition
$d_t = d_{t,1} + \cdots + d_{t,s}$ of the $d_t$~players of type~$t$.
Consider the strategy profile in which, for all $t\in {\cal T}$ and all $j\in\{1,\ldots,s\}$,
$d_{t,j}$ players of type~$i$ play the $j$-th strategy in~$S$.
Check whether this strategy profile is an $\epsilon$-Nash equilibrium.
\label{stepRS:bruteforce}

Return an $\epsilon$-Nash equilibrium if one is found.
\end{algorithm}

\begin{theorem}\label{thm:ptasfixedstrat}
For any constant $\epsilon>0$ and any constant $n$, Algorithm~\ref{ptasRS}
returns a $2\epsilon$-Nash equilibrium in time polynomial in $d$.
\end{theorem}

\begin{proof}
It is convenient to assume in the proof that $\epsilon$ (and hence $\delta$)
is the inverse of an integer. This can be assumed by rounding $\epsilon$ down
to the nearest such fraction $\epsilon'$, obtaining an $\epsilon'$-Nash equilibrium,
which is consequently an $\epsilon$-Nash equilibrium.

Consider first the case in which $\epsilon$ is the inverse of an integer.
Since costs lie in the range $[0,1]$ we can see that after
Step~\ref{stepRS:constantactions} is executed there are less than ${(\ell+1)}^n$ distinct player types.
We will find an $\epsilon$-Nash equilibrium of the rounded game,
which is a $2\epsilon$-Nash equilibrium of the original game since
each cost is only changed by at most an additive~$\epsilon$.
Note that $|{\cal T}|=O(1)$ as a function of $d$, the number of players.
Also, the number~$s$ of strategies in~$S$ constructed in Step~\ref{stepRS:all}
is $O(1)$ as a function of~$d$.
Since the number of $s$-compositions of $d_t$ (Step~\ref{stepRS:bruteforce}) is $O((d_t)^s)$,
the total number of mixed-strategy profiles considered in Step~\ref{stepRS:bruteforce} is
at most $O(\prod_{t\in{\cal T}} (d_t)^s)=O(d^{s|{\cal T}|})$ which is
polynomial in~$d$.

We now show that the algorithm always finds an $\epsilon$-Nash equilibrium in the last step.
First, consider the game with rounded cost vectors, as constructed
in Step~\ref{stepRS:constantactions}, and let ${\cal N}$ be a Nash equilibrium of this game.
Consider a $\delta$-rounding (as from Observation~\ref{obs:roundprob})
of each probability vector of ${\cal N}$ and notice
that such a probability vector is checked by the algorithm (or an equivalent
one is, in which the identities of players of the same type are swapped).
We next show that such a probability vector is an $\epsilon$-Nash equilibrium.
A player can be playing an action with a probability that differs by at most $\delta$
from the probability he should have used for his best response.
Thus for each of his actions, he can lose at most $\delta$ times the payoff
he gets for that action.
Since payoffs are upper-bounded by $1$, the maximum regret is
at most $n\delta=\epsilon$. The proof concludes by noting that the cost rounding
of Step \ref{stepRS:constantactions} implies an extra additive error of at most $\epsilon$.
\end{proof}

\subsubsection{FPTAS for constant number of players, many strategies}\label{sec:fptas}

Let $\G$ be a score-symmetric game with $d$ players having access to $n$ (shared) pure strategies.
In this section we view the number of players $d$ as a constant, and the number of
pure strategies $n$ is the parameter that governs the size of a game.
Let $\{1, \ldots, d\}$ denote the players.
We have $d$ prizes of values $u_1 \geq u_2 \geq \cdots \geq u_d$, where $u_1=1$ and $u_d=0$.

In this case, the expected payoff of player $i$ from playing $a_j$ is given
by the expected prize he gets minus his cost $c^i_j$ of $a_j$.
To define the expected prize that the player gets by playing $a_j$ we need some notation.
Let $B=\{-1, 0, 1\}$ and $v = (v_1, v_2, \ldots, v_d) \in B^d$ be a vector
that is defined with reference to some $a_j$ as follows.
$v_k$ gives information about the pure strategy played by player
$k\in\{1,\ldots,d\}$, where $v_k = 0$ means that player $k$ plays $a_j$,
$v_k=-1$ means that player $k$ plays one of the actions that are less competitive than $a_j$
(that is, one of $a_1,\ldots,a_{j-1}$), and finally,
$v_k=1$ means that $k$ plays one of the more competitive actions than $a_j$ (one of $a_{j+1},\ldots,a_n$).
Given vector $v \in B^d$ we let $v(1)$ be the number of $1$'s in $v$, i.e., $v(1) = |\{v_i: v_i=1\}|$;
similarly, $v(0)$ ($v(-1)$, respectively) denotes the number of $0$'s ($-1$'s, respectively) in $v$;
thus, $v(-1) + v(0) + v(1) = d$.

Now, let $x^{i_1}_j$ be the probability that player $i_1$ plays $a_j$,
and let $i_2, \ldots, i_d$ be the remaining players.
Observe that, if player $i_1$ plays purely $a_j$, i.e., $v_1 = 0$,
then prizes $u_1, u_2, \ldots, u_{v(1)}$ will be given to $v(1)$ players playing
more competitive actions than $a_j$, and prizes $u_{q+1}, u_{q+2}, \ldots, u_d$
are reserved for the $v(-1)$ players playing less competitive actions than $a_j$,
where $q = d - v(-1)$. Thus, the total value of prizes to be shared among the $v(0)$ players playing
$a_j$ (including player $i_1$) is $u_{v(1) + 1} + u_{v(1) + 2} + \cdots + u_q$.

We also denote by $\prob^{i}_j(-1) = \sum_{\ell=1}^{j-1} x^{i}_\ell$ and
$\prob^{i}_j(1) = \sum_{\ell=j+1}^n x^{i}_\ell$, the probabilities
of player $i$ playing actions that are less (respectively, more) competitive than $a_j$;
furthermore, let $\prob^{i}_j(0) = x^i_j$.
To cover degenerate cases we also assume that $\prob^{i}_1(-1) = 0$,
i.e., the probability of $i$ playing an action less competitive than $a_1$
(the weakest action) is zero; similarly, $\prob^{i}_n(1) = 0$.
Then the expected payoff of player $i_1$ for playing $a_j$, denoted as $\pi_j^{i_1}$, is given by
\begin{equation}\label{eq:payoff-2}
\pi_j^{i_1} = -c^{i_1}_j + \sum_{v = (v_1 = 0, v_2, \ldots, v_d) \in \{0\} \times B^{d-1}}
 \frac{u_{v(1) + 1} + u_{v(1) + 2} + \cdots + u_q}{v(0)} \cdot \left( \prod_{k=2}^d \prob^{i_k}_j(v_k) \right).
\end{equation}

To compute a Nash equilibrium, we need to find real values
$x^i_1,\ldots,x^i_n$, for $i \in \{1,2, \ldots, d\}$, that satisfy
\begin{equation}\label{eqn:probs-2}
x^i_j \geq 0~\forall i,j; ~~~~~~ \sum_j x^i_j = 1 ~~~~ i \in \{1,2, \ldots, d\}
\end{equation}
saying that for $i \in \{1,2, \ldots, d\}$, the values $\{x^i_j\}_j$ are a
probability distribution; for $i \in \{1,2, \ldots, d\}$ and $j>1$ the following should also hold:
\begin{equation}\label{eqn:pay-2}
\begin{array}{lll}
\hspace{-2pt} \pi^i_j > \max_{k=1,\ldots,j-1} \{\pi^i_k\} & \Longrightarrow & x^i_1=\ldots=x^i_{j-1}=0, \\
\hspace{-2pt}\pi^i_j < \max_{k=1,\ldots,j-1} \{\pi^i_k\} & \Longrightarrow & x^i_j=0.
\end{array}
\end{equation}

\begin{lemma}
The values $x^i_j$ satisfy~(\ref{eqn:probs-2}) and (\ref{eqn:pay-2}) if and only if they
are a Nash equilibrium.
\end{lemma}

\begin{proof}
The sets $\{x^i_j\}_j$, for $i \in \{1,2, \ldots, d\}$, are constrained by~(\ref{eqn:probs-2}) to
be probability distributions.

We claim that expressions~(\ref{eqn:pay-2}) are equivalent to the definition of Nash
equilibrium constraints~(\ref{eq:NEconstraints}).
If $a_j$ give player $i$ a higher payoff that all previous (weaker) strategies
then none of those may be in player $i$'s support.
Similarly, when $a_j$ gives a lower payoff than a weaker strategy, $a_j$ will not be in the support.
Note that if $a_j$ gives a higher payoff than the weaker actions
but a lower payoff than a stronger strategy $a_{j'}$
then the probability of playing $a_j$ will be set to $0$ when $\pi^i_{j'}$ is compared
with $\max_{k=1,\ldots,j,\ldots,{j'}-1} \{\pi^i_k\}$.
\end{proof}

Consequently we have reduced the problem to satisfying the
constraints~(\ref{eqn:probs-2}) and (\ref{eqn:pay-2}).
We now define variables in addition to $x$'s and $\pi$'s with the aim of expressing
(\ref{eqn:probs-2}) and (\ref{eqn:pay-2}) in terms of a constant number of ``local'' variables.
This is needed to define our FPTAS. Let $\sigma^i_j$ be the partial sum $\sum_{\ell=1}^j x^i_\ell$.
We can now express (\ref{eqn:probs-2}) as follows:
\begin{equation}\label{eqn:probs2-2}
\begin{array}{c}
\sigma^i_1=x^i_1  ~~~~~~~~ 0 \leq \sigma^i_j \leq 1 ~~~~~~~~ \sigma^i_{j-1}+x^i_j=\sigma^i_j \\
\sigma^i_n = 1 ~~~~~~~~ 0\leq x^i_j \leq 1,
\end{array}
\end{equation}
and observing $\prob^i_j(-1) = \sigma^i_{j-1}$, $\prob^i_j(0) = x^i_j$ and $\prob^i_j(1) = 1 - \sigma^i_j$,
we can now express $\pi^{i_1}_j$ in (\ref{eq:payoff-2}) only in terms of variables
$\sigma^{i_2}_{j-1}, \ldots, \sigma^{i_d}_{j-1}$,
$x^{i_2}_j, \ldots, x^{i_d}_j$, $\sigma^{i_2}_{j}, \ldots, \sigma^{i_d}_{j}$ as
\begin{equation}\label{eq:payoff2-2}
\pi_j^{i_1} = -c^{i_1}_j + \sum_{v \in \{0\} \times B^{d-1}}
                            \frac{u_{v(1) + 1} + \cdots + u_q}{v(0)} \cdot
       \Psi^{i_1}_j(v, \sigma^{i_2}_{j-1}, \ldots, \sigma^{i_d}_{j-1}, x^{i_2}_j, \ldots, x^{i_d}_j, \sigma^{i_2}_{j}, \ldots, \sigma^{i_d}_{j}),
\end{equation}
where function $\Psi^{i_1}_j(\cdot)$ is the product $\prod_{k=2}^d \prob^{i_k}_j(v_k)$
written in terms of these variables.
Observe that if $j=1$, then $\sigma^{i_2}_{j-1} = \cdots = \sigma^{i_d}_{j-1} = 0$,
and if $j=n$, then $\sigma^{i_2}_{j} = \cdots = \sigma^{i_d}_{j} = 1$;
thus, in these cases function $\Psi^{i_1}_j(\cdot)$ does not depend on these respective
variables.

Additionally, let $\alpha^i_j$ be the maximum expected payoff player $i$
can get by playing one of $a_1,\ldots,a_j$, i.e.,
$\alpha^i_j=\max_{k=1,\ldots,j} \{\pi^i_k\}$. We can now define
\begin{equation}\label{eq:maxpay-2}
\alpha^i_1 = \pi^i_1 ~~~~~~~~~~~ \alpha^i_j=\max\{\alpha^i_{j-1}, \pi^i_j\}
\end{equation}
and express (\ref{eqn:pay-2}) as follows:
\begin{equation}\label{eqn:pay2-2}
\begin{array}{l}
\pi^i_j > \alpha^i_{j-1} ~ \Longrightarrow ~ \sigma^i_{j-1}=0, \\
\pi^i_j < \alpha^i_{j-1} ~ \Longrightarrow ~ x^i_j=0.
\end{array}
\end{equation}

\begin{observation}\label{obs:rewriteNE-2}
The values $x^i_j$, $\sigma^i_j$, $\alpha^i_j$ and $\pi^i_j$
satisfy~(\ref{eqn:probs2-2},\ref{eq:payoff2-2},\ref{eq:maxpay-2},\ref{eqn:pay2-2})
if and only if the values $x^i_j$ are a Nash equilibrium.
\end{observation}

Now consider the sequence
\begin{equation*}\label{eqn:seq2}
{\cal S} = (\pi^1_j, \ldots, \pi^d_j, x^1_j, \ldots, x^d_j, \alpha^1_j, \ldots, \alpha^d_j, \sigma^1_j, \ldots, \sigma^d_j)_{j=1,\ldots,n}.
\end{equation*}
Constraints in (\ref{eqn:probs2-2}) involve $3$ variables that are at distance
at most $4d+1$ in $\cal S$ (namely, for $j>1$, $\sigma_{j-1}^i$ is followed by
$4d$ elements of ${\cal S}$ --including $x_j^i$-- and then by $\sigma^i_j$).
Constraints (\ref{eq:payoff2-2}) on the other hand involve variables that are
at distance at most $5d$ in $\cal S$. It is easy to check that the same happens
also for the other constraints and conclude then that the following holds.

\begin{observation}\label{obs:6-2}
For any $j=1,\ldots, n$, for each constraint in (\ref{eqn:probs2-2}),
(\ref{eq:payoff2-2}), (\ref{eq:maxpay-2}) and (\ref{eqn:pay2-2})
there are $5d$ consecutive elements of $\cal S$ that contains the quantities involved
in the constraint.
\end{observation}

\noindent{\bf The algorithm.} For $\epsilon>0$ according
to (\ref{eq:epsilonNEconstraints}) we relax the constraints of~(\ref{eqn:pay2-2}) as follows:
\begin{equation}\label{eqn:pay3-2}
\begin{array}{rcl}
\pi^i_j > \alpha^i_{j-1} + \epsilon & \Longrightarrow & \sigma^i_{j-1}=0, \\
\pi^i_j < \alpha^i_{j-1} - \epsilon & \Longrightarrow & x^i_j=0.
\end{array}
\end{equation}

Let ${\cal S}_i$ be the sequence of $5d$ consecutive elements of $\cal S$ that begins
at the $i$-th element of $\cal S$. Let ${\cal E}_i$ be the set of expressions
in~(\ref{eqn:probs2-2}), (\ref{eq:payoff2-2}), (\ref{eq:maxpay-2}) and
(\ref{eqn:pay3-2}) that relate elements of ${\cal S}_i$ with
each other; by Observation~\ref{obs:6-2} the union of the sets ${\cal E}_i$ is
all constraints~(\ref{eqn:probs2-2}), (\ref{eq:payoff2-2}), (\ref{eq:maxpay-2}) and (\ref{eqn:pay3-2}).
The algorithm (Algorithm~\ref{fptas-2}) works its way through the
sequence $\cal S$ left-to-right, and for each ${\cal S}_i$ identifies a subset of
$([0,1])^{5d}$ representing possible values of those quantities that form part of an
approximate Nash equilibrium. We call this subset $D_i$.
Then it sweeps through the sequence right-to-left identifying allowable values
for previous elements. The parameter $\epsilon$ controls quality of approximation.

\LinesNotNumbered
\begin{algorithm}[!htb]\label{fptas-2}

\caption{FPTAS for score-symmetric $d$-player games (where $d$ is constant)
having a variable number, $n$, of pure strategies and accuracy parameter $\epsilon>0$}

\nl For each player, round each cost $c^i_j$ down to the nearest non-negative
integer multiple of $\delta$, where $\delta = \frac{\epsilon}{4d^2 \cdot 3^{d}}$.

\nl For $1\leq i\leq 4dn-5d+1$, let $D_i$ be the set of all $5d$-dimensional
vectors that are calculated as follows:

 \Indp\Indp {\it a.} Generate all non-negative integer multiples of $\delta$ for the $x$'s in ${\cal S}_i$.

 {\it b.} For each collection of such values for the $x$'s use the constraints from (\ref{eqn:probs2-2}) in ${\cal E}_i$ to calculate the corresponding values of the $\sigma$'s in ${\cal S}_i$.

 {\it c.} For each collection of values of $x$'s and $\sigma$'s use the constraints from (\ref{eq:payoff2-2}) in ${\cal E}_i$ to calculate values of the $\pi$'s in ${\cal S}_i$.

 {\it d.} Calculate the $\alpha$'s values by applying the constraints from (\ref{eq:maxpay-2}) in ${\cal E}_i$.

 {\it e.} Among all computed vectors computed above, keep in $D_i$ 
 only those that 
 fulfill the constraints from (\ref{eqn:pay3-2}) in ${\cal E}_i$.

\Indm \Indm \nl For $i>1$ (in ascending order) keep in $D_i$ only vectors $s$ for which there is at least one
vector $s'$ in $D_{i-1}$ such that the first $5d-1$ entries of $s$ are the same as the
last $5d-1$ entries of $s'$.

\nl Let ${\bf s}_{4dn-5d+1}$ be a point in $D_{4dn-5d+1}$. For $1 \leq i < 4dn-5d+1$ (in
descending order) let ${\bf s}_i$ be a point in $D_i$ chosen so that its
last $5d-1$ coordinates are the first $5d-1$ coordinates of ${\bf s}_{i+1}$.

\nl Let ${\bf s}$ be the vector of length $4dn$ such that ${\bf s}_i$ is the $i$-th
sequence of $5d$ consecutive coordinates of ${\bf s}$.
Set $x^i_j$ to the entry of ${\bf s}$ that corresponds to the position of
$x^i_j$ in $\cal S$.\label{fptas:return}
\end{algorithm}

\begin{theorem}\label{thm:fptas2}
There is a FPTAS for computing $\epsilon$-Nash equilibria of competitiveness-based
ranking games with a constant number of players.
\end{theorem}

When the number of players is a constant, we may assume that the games is score-symmetric
by applying the reduction of Section~\ref{sec:score-symm}, which is then
achieved at the price of a constant-factor increase in the number $n$ of
strategies per player. We show that Algorithm~\ref{fptas-2} is indeed a FPTAS
for this class of games.
The proofs assume that $\epsilon$ is the inverse of an integer.
(Similarly to above, if this is not the case we simply run the algorithm with
an $\epsilon' < \epsilon$ which is inverse of an integer.)
Theorem~\ref{thm:fptas2} will follow from the following two propositions and
the subsequent observation about the runtime.

\begin{proposition}[Approximation guarantee]
If Algo\-rithm~\ref{fptas-2} finds a vector ${\bf s}$ in Step~\ref{fptas:return},
then the values $x^i_j$ correspond to a $(n+2)\epsilon$-Nash equilibrium.
\end{proposition}

\begin{proof}
The entries of ${\bf s}$ (including the subset corresponding to $x^i_j$)
satisfy (\ref{eqn:probs2-2}), (\ref{eq:payoff2-2}), (\ref{eq:maxpay-2})
and (\ref{eqn:pay3-2}), where~(\ref{eqn:pay3-2}) simply re\-writes the definition
of $\epsilon$-Nash equilibrium (\ref{eq:epsilonNEconstraints}) thus implying that
we are losing an additive $\epsilon$. Another additive loss smaller than $\epsilon$
is due to the cost rounding. Furthermore, we are restricting to probability
distributions whose values are non-negative integer multiples of $\delta$.
Thus, a player may be forced to play a strategy with a probability that differs by
at most $\delta$ from the probability of his best response.
This may impose an additional additive error of $n \delta < n \epsilon$ in the worst case
(this is because we have $n$ actions and on each of them the best response is
at most $\delta$ different, while the payoffs are upper bounded by $1$).
\end{proof}

\begin{proposition}[Correctness]
Algorithm \ref{fptas-2} always finds a vector ${\bf s}$ in Step~\ref{fptas:return}.
\end{proposition}

\begin{proof}
Consider a Nash equilibrium ${\cal N}$ and the associated vector ${\bf s}$.
Take a $\delta$-rounding $\tilde{x}$ ($\delta = \epsilon/(4d^2 \cdot 3^{d})$)
for each probability vector $x$ in ${\bf s}$.
Use (\ref{eqn:probs2-2}) to define the corresponding rounded vector $\tilde{\sigma}$.
Given $\tilde x$ and $\tilde \sigma$, use (\ref{eq:payoff2-2}) and (\ref{eq:maxpay-2})
to define values of $\tilde{\pi}$ and $\tilde{\alpha}$, respectively;
the newly-obtained vector of rounded values is denoted as $\tilde{\bf s}$.
Observe that $\tilde{\bf s}$ is considered by the algorithm.
It thus suffices to show that such a sequence satisfies all the constraints
that the algorithm imposes on the output.

By construction, (\ref{eqn:probs2-2}), (\ref{eq:payoff2-2}) and (\ref{eq:maxpay-2})
are fulfilled. For constraint (\ref{eqn:pay3-2}) we show that
$\tilde{\pi}^i_j > \tilde{\alpha}^i_{j-1} + \epsilon  \Rightarrow  \tilde{\sigma}^i_{j-1}=0$.
(Very similar arguments can be used to show the other condition of (\ref{eqn:pay3-2}).)

We will show first that $|y^i_j - \tilde{y}^i_j|<\epsilon/2$ for $y \in \{\pi,\alpha\}$.
We will only give the details for $y = \pi$ as the argument is similar in the other case.
Let us first focus on the product $\prod_{k=2}^d \prob^{i_k}_j(v_k)$ in (\ref{eq:payoff-2})
in which the factors $\sigma^{i_k}_{j-1}$, $1 - \sigma^{i_k}_{j}$ and $x^{i_k}_j$ are involved.
Now, the $\tilde{x}$'s are defined as $\delta$-roundings of the corresponding $x$'s
(according to Observation~\ref{obs:roundprob}).
The quantity $\tilde\sigma_j^i$ is the sum of the first $j$ entries
of $\tilde x^i$.
Thus, Observation~\ref{obs:roundprob} allows us to deduce that the rounded values
$\tilde{\sigma}^{i_k}_{j-1}$, $1 - \tilde{\sigma}^{i_k}_{j}$ and $\tilde{x}^{i_k}_j$
are less than $\delta$ away from $\sigma^{i_k}_{j-1}$, $1 - \sigma^{i_k}_{j}$ and
$x^{i_k}_j$ in absolute value, respectively. In other words,
$|\tilde{\prob}^{i_k}_j(v_k) - {\prob}^{i_k}_j(v_k)|< \delta$ for any $i_k$ and $j$.

Then we have
\[
\prod_{k=2}^d \tilde{\prob}^{i_k}_j(v_k)- \prod_{k=2}^d \prob^{i_k}_j(v_k)   <
\prod_{k=2}^d ( \prob^{i_k}_j(v_k)+\delta) - \prod_{k=2}^d \prob^{i_k}_j(v_k).
\]
%

An upper bound for the right-hand side is obtained by setting $\prob^{i_k}_j(v_k)=1$,
resulting in an upper bound of $(1+\delta)^{d-1}-1$. $\delta$ was chosen sufficiently small
to ensure that this is at most $\epsilon/(4\cdot 3^{d-1})$.
Now observe that the component
$\frac{u_{v(1) + 1} + u_{v(1) + 2} + \cdots + u_p}{v(0)}$ in (\ref{eq:payoff-2})
has value at most one and there are at most $3^{d-1}$ terms in the summation of
(\ref{eq:payoff-2}), which implies that the difference between the summations
for $\pi_j^{i_1}$ and $\tilde{\pi}_j^{i_1}$ is strictly smaller than $\epsilon/4$.
We have $|\pi^{i_1}_j - \tilde{\pi}^{i_1}_j|<\epsilon/2$.

Thus we have shown that $|y^i_j - \tilde{y}^i_j|<\epsilon/2$ for
$y \in \{\pi,\alpha\}$ and so $-\epsilon/2<y^i_j - \tilde{y}^i_j<\epsilon/2$
for $y \in \{\pi,\alpha\}$.
Therefore $\tilde{\pi}^i_j>\tilde{\alpha}^i_{j-1}+\epsilon$ implies that
${\pi}^i_j>{\alpha}^i_{j-1}$ and as $\bf s$ is a Nash equilibrium,
by Observation~\ref{obs:rewriteNE-2} and (\ref{eqn:pay2-2}), we have
$\sigma^i_{j-1}=0$. But then by the way we define $\tilde{\sigma}$'s
we have that $\tilde{\sigma}^i_{j-1}=0$.
\end{proof}

\begin{paragraph}{Runtime.}
Since the values of $\sigma$'s, $\pi$'s and $\alpha$'s
are calculated applying (\ref{eqn:probs2-2}), (\ref{eq:payoff2-2}) and
(\ref{eq:maxpay-2}) respectively for given $x$'s, the sets $D_i$ are of
size $O((1/\delta)^{d}) = O((1/\epsilon)^{d}) $, so for constant
$d$ the runtime of the algorithm is indeed polynomial in $n$ and $1/\epsilon$,
as required for a FPTAS.
\end{paragraph}

\begin{paragraph}{Remarks.}
Algorithm~\ref{fptas-2} is somewhat similar to the algorithm of~\cite{KLS} for
solving tree-structured graphical games. They give a similar
forward-and-backward dynamic programming approach to solving these
games; their algorithm takes exponential time for exact
equilibria~\cite{EGG} but a similar quantisation of real-valued
payoffs leads to a FPTAS.
\end{paragraph}

\section{Conclusions and further work}

Our FPTAS can be used to compute \emph{exact} equilibria in certain cases.
When a game with constantly-many players has payoffs that are multiples of
some $\epsilon>0$ then we can compute exact Nash equilibria in time polynomial
in the size of the input and $1/\epsilon$ by simply using the FPTAS.
This observation raises the open problem of determining whether there
is a polynomial-time algorithm  for solving 2-player
(competitiveness-based ranking) games in general when ties are possible
and the prize is shared in the event of a tie.

Several other concrete open problems have been raised by the current
results,  for example,  fully quantifying the complexity of computing
Nash equilibria for competitiveness-based ranking games.
Also, in situations where multiple equilibria may exist,
we would like to know whether a specific equilibrium is selected by
some natural decentralized dynamic process.

\bigskip \noindent {\bf Acknowledgements.} We thank David Gill,
Milan Vojnovik and Yoram Bachrach for pointers to related work.

\end{document}